%% file: arxiv.tex
\documentclass[aps, prl, reprint, letterpaper, twocolumn, notitlepage, superscriptaddress]{revtex4-2}

\usepackage[charter, cal=cmcal, sfscaled=false]{mathdesign}
\usepackage[hidelinks, colorlinks=true, citecolor=blue, urlcolor=., linkcolor=.]{hyperref}
\usepackage{graphicx, braket, mathtools, appendix, amsmath, listings, xcolor, printlen, tikz, amsthm, ragged2e, comment, ifthen}
\usepackage[shortlabels]{enumitem} 
\usepackage[parfill]{parskip} 
\usepackage[mode=buildnew]{standalone} 

\usetikzlibrary{decorations.pathreplacing, arrows.meta, calc, backgrounds, fit, quantikz2} 

\newboolean{arxiv}
\setboolean{arxiv}{true} 

\newtheorem{lemma}{Lemma} 
\newtheorem{example}{Example} 

\justifying
\begin{document}
		
\title{An angular momentum approach to quantum insertion errors}

\author{Lewis Bulled}
\email{lbulled1@sheffield.ac.uk}

\author{Yingkai Ouyang}
\email{y.ouyang@sheffield.ac.uk}

\affiliation{School of Mathematical and Physical Sciences, University of Sheffield, Sheffield, S3 7RH, United Kingdom}

\date{\today}

\begin{abstract}
    \noindent Quantum insertion errors are a class of errors that increase the number of qubits in a quantum system. Despite a wealth of research on classical insertion errors, there has been limited progress towards a general framework for correcting quantum insertion errors. We detail a quantum error correction protocol that can correct single insertion errors on a class of gapped permutation-invariant codes. We provide a simple two-stage syndrome extraction protocol that yields a two-bit syndrome, by measuring the total angular momentum and its projection along the $z$-axis (modulo the code gap) of the post-insertion state. We demonstrate that these measurements project the state onto a new codespace, and we detail a teleportation protocol to map the projected state back to a permutation-invariant code on the desired number of qubits.
\end{abstract}

\maketitle

\input{main}

\onecolumngrid
\section*{Appendix}
\label{Appendix}
\input{supp}

\twocolumngrid


%

\end{document}

%% file: main.tex
\section{Introduction}
\label{Sec:Introduction}

\noindent There is a real need for quantum technologies to realise large-scale, fault-tolerant quantum computing. It is well-known that quantum error correction (QEC) plays an important role in this endeavour, due to the presence of external noise in quantum computers. As such, QEC literature is vast, with a focus on single-qubit errors such as Pauli errors, amplitude damping, depolarisation and decoherence. A class of quantum errors that has only recently been considered is so-called synchronisation errors, that is, errors which alter the number of qubits in a quantum system. These include erasure, deletion and insertion errors; erasures delete qubits at known locations, whereas deletions (insertions) delete (insert) qubits at unknown locations according to some probability distribution.

\noindent Synchronisation errors often occur in optical quantum communication. When states are transmitted over a lossy quantum channel, some level of qubit loss is expected due to erasure and/or deletion errors \cite{Ouyang_2021}; equally, insertion errors may be exhibited due to environmental noise or hardware imperfections. Moreover, insertion errors occur in quantum key distribution where eavesdroppers conduct side-channel attacks, such as trojan-horse attacks \cite{Vinay_Kok_2018, Jain_2015, Navarrete_2022, Nasedkin_2023}. As such, synchronisation error QEC is paramount to the future establishment of a quantum network. Despite this, is unclear how to correct synchronisation errors; since the dimension of the Hilbert space is altered, conventional QEC techniques involving stabiliser or surface codes cannot be used.

\noindent Classically, research on insertion errors dates back more than sixty years. Levenshtein demonstrated an equivalence between the correctability of classical deletions and insertions \cite{Levenshtein_1966}, that is, a classical code capable of correcting $t$ deletions can also correct $t$ insertions (and vice versa.) Research on classical insertion-deletion (``insdel'') codes is reasonably well-studied \cite{Schulman_Zuckerman_1997, Haeupler_2017, Duc_Liu_2021}, with applications in racetrack memories and DNA storage \cite{Chee_2017, Buschmann_2013}. Despite this, limited progress has been made towards the quantum analogue.

\noindent Whilst research on quantum deletion codes has gained some traction in recent years \cite{Nakayama_Hagiwara_Jan2020, Hagiwara_Nakayama_Apr2020, Shibayama_Hagiwara_2021, Ouyang_2021}, there remains a paucity of research on quantum insertion codes. Although quantum insdel channels have been formulated \cite{Leahy_2019}, as yet there exists no general framework for correcting insertion errors in the literature. Hagiwara provided the first instance of a quantum insertion code \cite{Hagiwara_2021}, whilst Shibayama and Hagiwara provided a class of codes capable of correcting both single insertions and deletions \cite{Shibayama_Hagiwara_2022}. Levenshtein's classical insdel equivalence \cite{Levenshtein_1966} remains an open problem for quantum errors, so naturally much of the focus has been towards a quantum insdel equivalence. Shibayama and Ouyang proved such an equivalence for separable insertions \cite{Shibayama_Ouyang_2021}, with Shibayama then extending this to arbitrary single insertions \cite{Shibayama_2025}. Ouyang and Brennen adopted a different QEC approach in \cite{Ouyang_Brennen_2025}, proposing a syndrome extraction protocol that measures the total angular momentum (AM) of consecutive subsets of qubits. The authors demonstrated that this corrects all single-qubit and deletion errors, though insertion errors were not considered.

\noindent We present the first instance of a general quantum insertion error-correcting protocol, which can be used to correct single insertion errors on gnu codes with a code gap $g$. We adopt the same AM approach as in \cite{Ouyang_Brennen_2025}, in contrast to the combinatorial arguments utilised in \cite{Shibayama_Hagiwara_2021, Shibayama_Hagiwara_2022}. Syndrome extraction proceeds by measuring the total AM and its projection along the $z$-axis modulo $g$, via the operators $J^2$ and $J^z$ (mod $g$) respectively. These measurements yield a two-bit syndrome $(j,w)$ and project the post-insertion state onto a new codespace. We detail a teleportation protocol to map the projected state back to a permutation-invariant (PI) code on the desired number of qubits. Our QEC protocol is a simple algorithm that requires only two measurements, and decoding is straightforward since our syndrome comprises just two classical bits.

\noindent Furthermore, all elements of our protocol (state preparation, measurements and QEC) can be implemented using geometric phase gates (GPGs) \cite{Luis_2001, Xiaoguang_Zanardi_2002}, as in \cite{Johnsson_2020}. GPGs rely on a dispersive coupling of the qubits with a bosonic mode, and require only four native operations. First is the initialisation of the mode, which is achievable with a laser. Second is a coherent on-off dispersive coupling of all qubits to that mode, such as in cavity quantum electrodynamics architectures \cite{Eickbusch_2022, Wang_2024}. Third, we require displacements of the mode, and fourth, homodyne detection, both of which have implementations on various platforms \cite{Maioli_2005, Du_Vogt_Li_2023}. By moving beyond a universal gate decomposition (such as the Clifford + $T$ gateset), GPGs sidestep the need for individual qubit addressability, and make use of existing bosonic manipulations. This offers an efficient route to realising our protocol on near-term quantum devices.

\section{Background}
\label{Sec:Background}

\noindent \textbf{Permutation-invariant codes.} We focus on PI codes, a class of quantum codes that are particularly effective for synchronisation error QEC. Proposed by Ruskai \cite{Ruskai_2000}, PI codes are contained within the symmetric space; that is, the subspace of $\mathbb{C}^{2 \otimes N}$ such that interchanging any pair of qubits leaves resultant states unchanged. The symmetric space is spanned by Dicke states $\Ket{D^N_k}$, where the Dicke state of weight $k$ is given by
\begin{equation}
    \label{Dicke state}
     \Ket{D^N_k} = \frac{1}{\sqrt{\binom{N}{k}}} \sum_{\substack{x \in \{0,1\}^N, \\ \text{wt}(x) = k}} \ket{x}.
\end{equation}
Such states can be viewed as normalised superpositions of permutations of $N$ spin-$\frac{1}{2}$ particles with $k$ spin-up particles and $N-k$ spin-down particles \cite{Ouyang_2014}. For example, the Dicke state $\ket{D^3_1}$ is given by
\begin{equation}
    \Ket{D^3_1} = \frac{\Ket{\uparrow \downarrow \downarrow} + \Ket{\downarrow \uparrow \downarrow} + \Ket{\downarrow \downarrow \uparrow}}{\sqrt{3}}. 
\end{equation}
There exists a wealth of literature on PI codes spanning more than two decades, cf. \cite{Ruskai_2000, Pollatsek_Ruskai_2004, Ouyang_2014, Ouyang_Fitzsimons_2016, Ouyang_2017, Hagiwara_Nakayama_Jan2020, Aydin_2024, Ouyang_Jing_Brennen_2025}.

\noindent The first step of any QEC protocol is encoding. For PI codes, we encode a single qubit $c_0 \Ket{0} + c_1 \Ket{1}$ into a symmetric $N$-qubit state using GPGs \cite{Johnsson_2020}, though one can also use a universal gate decomposition \cite{Plesch_2011, Bartschi_2019}. For $0 \leq k \leq N$,
\begin{equation}
    \label{N-qubit PI state}
    \Ket{\psi_N} = \sum_k \beta_k \Ket{D^{N}_k} \equiv c_0 \Ket{0_L} + c_1 \Ket{1_L},
\end{equation}
where $|c_0|^2 + |c_1|^2 = 1$ and $\Ket{x_L}$ are logical codewords for $x \in \{0,1\}$, given by a linear combination of Dicke states. We restrict ourselves to gnu codes henceforth \cite{Ouyang_2014}, an infinite family of PI codes on $N=gnu$ qubits. Here, $g, n \geq 2$ are the code gap and code occupancy respectively, and $u \geq 1$ is a scaling parameter that determines the code length \cite{Ouyang_2021}. The code distance is given by $\text{min}(g,n)$, with logical codewords
\begin{equation}
    \label{gnu logical codewords}
    \Ket{x_L} = 2^{-\frac{n-1}{2}} \sum_{i \equiv x \, (\text{mod} \, 2)} \sqrt{\binom{n}{i}} \, \Ket{D^N_{gi}}
\end{equation}
for $x \in \{0,1\}$ and $0 \leq i \leq nu$. The particular PI code depends on the choice of $\beta_k$ in \eqref{N-qubit PI state}; for gnu codes we have
\begin{equation}
    \label{Definition of beta}
    \beta_{gi} \coloneqq 
    \begin{dcases}
        2^{-\frac{n-1}{2}} c_0 \sqrt{\binom{n}{i}}& \text{if $i$ even}, \\
        2^{-\frac{n-1}{2}} c_1 \sqrt{\binom{n}{i}} & \text{if $i$ odd}, 
    \end{dcases}
\end{equation}
and $\beta_k = 0$ for all $k \neq gi$. Notable examples of gnu codes include the four-qubit code \cite{Hagiwara_Nakayama_Apr2020} with $g=n=2, u=1$ \cite{Ouyang_2021} and the Ruskai code \cite{Ruskai_2000} with $g=n=3, u=1$ \cite{Ouyang_2021}.

\noindent \textbf{Insertion errors.} Insertion errors occur when additional qubits appear amongst the logical state \eqref{N-qubit PI state}. We consider the case of a single insertion error only, where the qubit count evolves as $N \rightarrow N+1$. In this case, the inserted qubit can appear in one of $N+1$ possible positions, denoted by $a$ for $0 \leq a \leq N$, though this position is unknown. An illustrative diagram depicting such an error is given in Fig. \ref{Fig:Insertion error}. 

\noindent We represent the logical state \eqref{N-qubit PI state} as a spin-$\frac{N}{2}$ particle and the insertion error \eqref{Single insertion state} as a spin-$\frac{1}{2}$ particle, since both are two-level systems alike qubits. In this way, we can reframe an insertion error as an AM coupling problem, with the advantage of leveraging the laws of AM coupling, a well-understood phenomenon in quantum mechanics.

\noindent The logical state has maximal AM $j=\frac{N}{2}$ and magnetic numbers $m = \frac{N}{2}, \ldots, -\frac{N}{2}$. Now consider a single pure state insertion error on \eqref{N-qubit PI state}, with $j = \frac{1}{2}$ and $m = \pm\frac{1}{2}$, of the form
\begin{equation}
    \label{Single insertion state}
    \Ket{\psi_1} = \sum_m \alpha_m \ket{1/2, m} \equiv v_0 \Ket{0} + v_1 \Ket{1},
\end{equation} 
where $|v_0|^2 + |v_1|^2 = 1$. This results in a state with total AM $j=\frac{N \pm 1}{2}$ and magnetic numbers $m = \frac{N+1}{2}, \cdots, -\frac{N+1}{2}$, which correspond to the eigenvalues of $\hat{J}^z$ as defined in \eqref{Angular momentum operators}. Formally, we write the post-insertion state as $\Ket{\Psi^a} \coloneqq \pi_a \, \bigl( \Ket{\psi_N} \Ket{\psi_1} \bigr)$, where $\pi_a \in S_{N+1}$ is a permutation on $N+1$ qubits representing an insertion error at position $a$.

\begin{figure}
    \centering
    \includegraphics[scale=0.63]{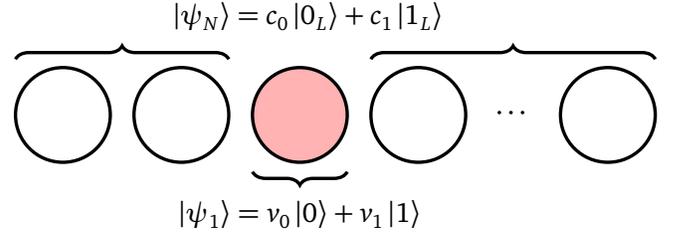}
    \caption[Insertion error]{Depiction of a single quantum insertion error $\ket{\psi_1}$ on the logical state $\ket{\psi_N}$ in position 2. Insertion position is labelled by the qubit directly to the left, with the exception of position 0.}
    \label{Fig:Insertion error}
\end{figure}

\section{Syndrome extraction}
\label{Sec:Syndrome extraction}

\noindent \textbf{Measurement of $\hat{J}^2$.} Syndrome extraction begins by measuring the total AM operator
\begin{equation}
    \label{J^2 operator}
    \hat{J}^2 = \left( \hat{J}^x \right)^2 + \left( \hat{J}^y \right)^2 + \left( \hat{J}^z \right)^2
\end{equation}
on the post-insertion state, where for $1 \leq \ell \leq N+1$,
\begin{equation}
    \label{Angular momentum operators}
    \hat{J}^x = \frac{1}{2} \sum_\ell X_\ell, \quad \hat{J}^y = \frac{1}{2} \sum_\ell Y_\ell, \quad \hat{J}^z = \frac{1}{2} \sum_\ell Z_\ell.
\end{equation}
Here, $X_\ell, Y_\ell, Z_\ell$ are the usual Pauli operators acting on the $\ell^{\text{th}}$ qubit. The projector onto the space with total AM $j$ is given by
\begin{equation}
    \label{J^2 projectors}
    \mathcal{P}_j = \sum_{m, p} \Ket{j,m}_p \Bra{j,m}_p,
\end{equation}
where $\big\{ \ket{j,m}_p \big\}_{j,m,p}$ is a sequentially coupled basis (SCB). We construct such an SCB by coupling the spins $j_i = \frac{1}{2}$ in the following manner
\begin{equation}
    \label{Coupling path}
    \left( j_1, j_2 \right) \rightarrow j_{[2]}, \, \left( j_{[2]}, j_3 \right) \rightarrow j_{[3]}, \, \cdots, \, \left( j_{[N]}, j_{N+1} \right) \rightarrow j,
\end{equation}
where $[i] \coloneqq \{1, \cdots, i\}$. We choose the SCB where we couple $j_i$ as in \eqref{Coupling path}, but one could construct an SCB via any other coupling path. Note that the subscript $p$ in \eqref{J^2 projectors} denotes the coupling path $p = (j_1, j_{[2]}, \cdots, j_{[N]}, j)$, a label necessary to distinguish between degenerate AM eigenstates for fixed $j$. There exist several other ways to represent this notion, including standard Young tableau, Yamanouchi symbols and binary trees, cf. \cite{Havlicek_Strelchuk_2019, Wills_Strelchuk_2024}.

\noindent After measuring the projector \eqref{J^2 projectors}, $\ket{\Psi^a}$ is projected onto the Schur-Weyl basis \cite{Bacon_Chuang_Harrow_2006}. More precisely, $\ket{\Psi^a_j} \in \text{span}\{ \mathcal{W}_{j,w} \}$ where $\mathcal{W}_{j,w} \coloneqq \left\{ \ket{j,m}_p \mid j = \frac{N \pm 1}{2}, -j \leq m \leq j \right\}$ is the subspace of $\mathbb{C}^{2 \otimes (N+1)}$ with total AM $j$. This is a direct consequence of Schur-Weyl duality  \cite{Harrow_2005}, which decomposes the Hilbert space into irreducible representations of the symmetric group $S_{N+1}$ and special unitary group $SU(2)$. Schur-Weyl duality has been routinely utilised in quantum information theory, cf. \cite{Havlicek_Strelchuk_2018, East_2023}. Here, we utilise the block-diagonal structure of the Hilbert space in the Schur-Weyl basis, since measurement dephases the post-insertion state.

\noindent \textbf{Modular measurement of $\hat{J}^z$.} Syndrome extraction proceeds by measuring $\hat{J}^z$ on the post-insertion state modulo the code gap $g$. The projector onto the subspace with magnetic number $m$ can be written as
\begin{equation}
    \label{Projector onto subspace}
    \mathcal{Q}_m = \sum_w \mathcal{P}^{w}_j, 
\end{equation}
where $0 \leq w \leq g-1$ and
\begin{equation}
    \label{General projector}
    \mathcal{P}^w_j = \sum_p \Pi^{w}_{j, p}.
\end{equation}
For a given $p$ and $0 \leq i \leq \lfloor \frac{2j-w}{g} \rfloor$, the projector
\begin{equation}
    \label{Projector onto specific spin path}
    \Pi^w_{j, p} = \sum_i \Ket{j, gi+w-j}_p \Bra{j, gi+w-j}_p
\end{equation}
projects onto $\text{span} \{ \ket{j, m}_p \mid j+m \equiv w \,\, (\text{mod} \, g) \}$. Thus for all $p$, $\mathcal{P}^w_j$ projects onto the space where the sum of the total AM and magnetic numbers are congruent to $w$ modulo $g$.

\noindent \textbf{Syndrome extraction protocol.} We now detail a two-stage syndrome extraction protocol that yields two `bits' of classical information, $(j,w)$. The syndrome is effectively two bits since $j$ takes only two values, and although $w$ takes a possible $g$ values, we find that only two such values lead to a non-zero projection. We measure $\hat{J}^2$ on the post-insertion state $\ket{\Psi^a}$ to obtain $j=\frac{N \pm 1}{2}$, and (unnormalised) post-measurement states $\ket{\tilde \Psi^a_j} \coloneqq \mathcal{P}_j \Ket{\Psi^a}$. We then measure $\hat{J}^z$ (mod $g$) on the normalised state $\ket{\Psi^a_j} = \ket{\tilde \Psi^a_j} / {\sqrt{\braket{\tilde \Psi^a_j | \tilde \Psi^a_j}}}$ to obtain $w=0, \ldots, g-1$ and (unnormalised) post-measurement states $\ket{\tilde \Psi^{a,w}_j} \coloneqq \mathcal{P}^w_j \ket{\Psi^a_j}$. As a result, we obtain the projected states $\ket{\Psi^{a,w}_j} = \ket{\tilde \Psi^{a,w}_j} / {\sqrt{\braket{\tilde \Psi^{a,w}_j | \tilde \Psi^{a,w}_j}}}$ and one of four possible syndromes 
\begin{align}
    \notag \left( j=\frac{N+1}{2}, \, w=0 \right), \quad &\left( j=\frac{N+1}{2}, \, w=1 \right), \\
    \left( j=\frac{N-1}{2}, \, w=0 \right), \quad &\left( j=\frac{N-1}{2}, \, w=g-1 \right).
\end{align}
A visual summary of our syndrome extraction protocol is given in Fig. \ref{Fig:Syndrome extraction}.

\begin{figure}
    \centering
    \includegraphics[scale=0.63]{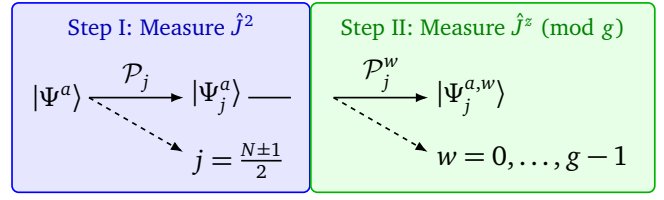}
    \caption[Syndrome extraction]{Summary of our two-stage syndrome extraction protocol via an AM approach. We measure $\hat{J}^2$ on $\ket{\Psi^a}$ to obtain $j=\frac{N \pm 1}{2}$, and $\ket{\Psi^a_j} \propto \mathcal{P}_j \Ket{\Psi^a}$. We then measure $\hat{J}^z$ (mod $g$) on $\ket{\Psi^a_j}$ to obtain $w=0, \ldots, g-1$ and $\ket{\Psi^{a,w}_j} \propto \mathcal{P}^w_j \ket{\Psi^a_j}$. Such measurements can be efficiently implemented using GPGs \cite{Johnsson_2020}.}
    \label{Fig:Syndrome extraction}
\end{figure}

\section{Main result} 
\label{Sec:Main result}

\noindent We now explain how the above syndrome extraction protocol allows for insertion error QEC on PI codes. Suppose we measure $\hat{J}^2$ and obtain $j = \frac{N+1}{2}$, i.e. we project onto the symmetric space. Then measuring the projectors $\mathcal{P}^w_\frac{N+1}{2}$ on $\ket{\Psi_\frac{N+1}{2}}$ for $w = 0, \ldots, g-1$ yields
\begin{equation*}
    \label{Symmetric projection}
    \Ket{\tilde \Psi^w_\frac{N+1}{2}} =
    \begin{dcases}
        \tilde v_0' \, \left( c_0 \Ket{\tilde 0^0_\frac{N+1}{2}} + c_1 \Ket{\tilde 1^0_\frac{N+1}{2}} \right), & w=0; \\ 
        \tilde v_1' \, \left( c_0 \Ket{\tilde 0^1_\frac{N+1}{2}} + c_1 \Ket{\tilde 1^1_\frac{N+1}{2}} \right), & w=1; \\
        0, & 2 \leq w \leq g-1,
    \end{dcases}
\end{equation*}
where the probability of projecting onto a particular $w$, $|\tilde v_w'|^2$, is proportional to the probability of a $\ket{w}$-insertion, $|v_w|^2$. Here, $\big\{ \ket{\tilde 0^w_{\frac{N+1}{2}}}, \ket{\tilde 1^w_{\frac{N+1}{2}}} \big\}$ are orthogonal for $w \in \{0,1\}$, and are similar to the logical codewords of \eqref{gnu logical codewords} but on one additional qubit and with the binomial-type coefficients \eqref{Definition of beta} rescaled by CG coefficients. We observe a similar phenomenon for the mixed symmetry space with $j=\frac{N-1}{2}$,
\begin{equation*}
    \label{Mixed symmetry projection}
    \Ket{\tilde \Psi^{a,w}_\frac{N-1}{2}} =
    \begin{dcases}
        \tilde v_{a,1} \left( c_0 \Ket{\tilde 0^{a,0}_\frac{N-1}{2}} + c_1 \Ket{\tilde 1^{a,0}_\frac{N-1}{2}} \right), & w=0; \\
        0,  &1 \leq w \leq g-2; \\
        \tilde v_{a,0} \left( c_0 \Ket{\tilde 0^{a,g-1}_\frac{N-1}{2}} + c_1 \Ket{\tilde 1^{a,g-1}_\frac{N-1}{2}} \right), & w=g-1. \\
    \end{dcases}
\end{equation*}
As with the symmetric space, $|\tilde v_{a,w}|^2$ are proportional to $|v_w|^2$, and $\ket{\tilde x^{a,w}_{\frac{N-1}{2}}}$ are similar in structure to $\ket{\tilde x^w_{\frac{N+1}{2}}}$ for $x \in \{0,1\}$, but with more complicated rescaling. These projections can be obtained by first expressing $\ket{\Psi^a}$ in the SCB given by \eqref{Coupling path}, with the help of Clebsch-Gordan (CG) coefficients. One can then derive the post-measurement states after measuring $\hat{J}^2$ on $\ket{\Psi^a}$, followed by $\hat{J}^z$ modulo the code gap $g$ on $\ket{\Psi^a_j}$. Full details can be found in the \ifthenelse{\boolean{arxiv}}{Appendix}{Supplemental Material \cite{supp}}.

\noindent Upon inspection, we cannot identify $\ket{\tilde x^{a,w}_j}$ as logical codewords since they are unnormalised; furthermore, they do not necessarily have the same norm. Remarkably, we show in Lem. \ref{Norm-preserving lemma} that the norms of $\ket{\tilde x^{a,w}_j}$ are in fact equal. \newline

\begin{lemma}  
    \label{Norm-preserving lemma}
    Let $j=\frac{N \pm 1}{2}$ and $w=0, \ldots, g-1$ be fixed. Then the inner products $\braket{\tilde x^{a,w}_j | \tilde x^{a,w}_j}$ are equal for $x \in \{0,1\}$.
\end{lemma}
\begin{proof}
    Omitted for brevity.
\end{proof}

\noindent A proof of Lem. \ref{Norm-preserving lemma} can be found in the \ifthenelse{\boolean{arxiv}}{Appendix}{Supplemental Material \cite{supp}}. Our proof utilises \cite[Thm. 1]{O_Hara_2001} to express certain CG coefficients in terms of binomial coefficients, a recursion relation for CG coefficients \cite[Eq. (3.369)]{Sakurai_Napolitano_2020}, a well-known combinatorial identity by Vandermonde and results for binomial sums, including \cite[Lem. 1]{Ouyang_2014}. Lem. \ref{Norm-preserving lemma} above implies that for fixed $w$ and $x \in \{0,1\}$, \textbf{the normalised states $\ket{x^{a,w}_j} = \ket{\tilde x^{a,w}_j} / \sqrt{\braket{\tilde x^{a,w}_j | \tilde x^{a,w}_j}} $ form an orthonormal basis for the codespace of a quantum code on $N+1$ qubits}. Thus $\big\{ \ket{0^w_{\frac{N+1}{2}}}, \ket{1^w_{\frac{N+1}{2}}} \big\}$ represent the logical codewords of a PI code, and $\big\{ \ket{0^{a,w}_{\frac{N-1}{2}}}, \ket{1^{a,w}_{\frac{N-1}{2}}} \big\}$ represent the logical codewords of a spin code. We provide an example of our protocol below. \newline

\begin{example}
    The four-qubit code \cite{Hagiwara_Nakayama_Apr2020} is a gnu code with $g=n=2$, $u=1$ \cite{Ouyang_2021}. Thus the logical state is given by
    \begin{equation}
        \ket{\psi_4} = c_0 \left( \frac{\ket{D^4_0} + \ket{D^4_4}}{\sqrt{2}} \right) + c_1 \ket{D^4_2}.
    \end{equation}
    A single insertion error occurs on $\ket{\psi_4}$ and we obtain the post-insertion state $\ket{\Psi}$. If we measure $\hat{J}^2$, $\hat{J}^z$ (mod 2) and acquire the syndrome $\left( \frac{5}{2}, w \right)$ for $w \in \{0,1\}$, we obtain
    \begin{align}
        \notag \Ket{\Psi^0_{5/2}} &= c_0 \left( \frac{\sqrt{5}\ket{D^5_0} + \ket{D^5_4}}{\sqrt{6}} \right) + c_1 \ket{D^5_2}, \\
        \Ket{\Psi^1_{5/2}} &= c_0 \left( \frac{\ket{D^5_1} + \sqrt{5}\ket{D^5_5}}{\sqrt{6}} \right) + c_1 \ket{D^5_3}.
    \end{align}
    On the other hand, suppose we measure $\hat{J}^2$, $\hat{J}^z$ (mod 2) and acquire the syndrome $\left( \frac{3}{2}, w \right)$ for $w \in \{0,1\}$. Then for all $a=0, \ldots, 4$,
    \begin{align}
        \notag \Ket{\Psi^{a,0}_{3/2}} &= \sum_p d_{a,p} \left( c_0 \Ket{\frac{3}{2}, \frac{3}{2}}_p + c_1\Ket{\frac{3}{2}, -\frac{1}{2}}_p \, \right), \\
        \Ket{\Psi^{a,1}_{3/2}} &= \sum_p d_{a,p} \left( c_0 \Ket{\frac{3}{2}, -\frac{3}{2}}_p + c_1 \Ket{\frac{3}{2}, \frac{1}{2}}_p \, \right),
    \end{align}
    for some $d_{a,p} \in \mathbb{C}$ such that $\sum_p |d_{a,p}|^2 = 1$.
\end{example}

\section{Recovery} 
\label{Sec:Recovery}

\noindent We now discuss how to implement QEC for single insertion errors on gnu codes. There are a number of possible approaches to map the projected state $\ket{\Psi^{a,w}_j}$ to a PI code; for example, using techniques described in \cite{Ouyang_Brennen_2025}, such as the Gram-Schmidt procedure, quantum Schur transform or teleportation. We choose to map the state back to the symmetric space via teleportation, with the help of a PI ancilla as in \cite{Ouyang_Brennen_2025}. To do this, we must implement a conditional logical-X gate on the code, which is possible for odd code gaps $g$. We first define such a logical-X gate 
\begin{equation}
    \label{Logical X}
    X_L: \ket{j, m}_p \rightarrow \ket{j, -m}_p,
\end{equation}
as well as the logical-CNOT gate 
\begin{equation}
    \label{Logical CNOT}
    C_A X_B: \ket{j, m}_p \ket{j, m'}_{p'} \rightarrow \ket{j, m}_p (X_L)^m \ket{j, m'}_{p'}.
\end{equation}
Physically, the logical gates \eqref{Logical X}, \eqref{Logical CNOT} can be implemented using GPGs \cite{Johnsson_2020}. 

\noindent Our teleportation protocol proceeds as follows. In register A, we prepare a logical ancilla $\ket{+_L}$ in a gnu code with the same gap $g$ as $\ket{\psi_N}$; meanwhile, we prepare the spin code $\ket{\Psi^{a,w}_j}$ in register B. We then implement the logical-CNOT gate $C_A X_B$ between the two codes, with control on register A and target on register B. Next, we measure register B in the logical-Z basis of the new codespace, that is, $\text{span} \big\{ \ket{0^{a,w}_j}\bra{0^{a,w}_j}, \ket{1^{a,w}_j}\bra{1^{a,w}_j} \big\}$. Finally, we implement the logical-X gate $X_L$ on register A, conditional on the classical measurement outcome of register B. The effect of this protocol is to teleport the spin code to a PI code on the desired number of qubits. A quantum circuit is provided in Fig. \ref{Fig:Teleportation protocol}, which is a simple adaptation of \cite[Eq. 7]{Zhou_2000}.

\begin{figure}
    \centering
    \includegraphics[scale=1.35]{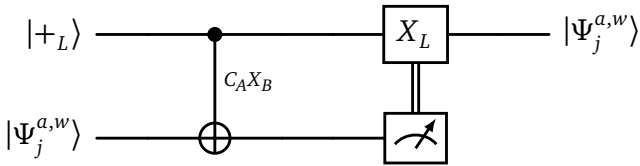}
    \caption[Teleportation circuit]{Quantum circuit for our teleportation protocol. The logical controlled-NOT gate $C_A X_B$ acts identically on $\ket{\Psi^{a,w}_j}$ for all $a=0, \ldots, N$, and can be implemented using GPGs \cite{Johnsson_2020}.}
    \label{Fig:Teleportation protocol}
\end{figure}

\noindent One benefit of our QEC protocol is that recovery has great flexibility. If we measure $\hat{J}^2$ and obtain $j = \frac{N+1}{2}$, then the projected state remains in the symmetric space on $N+1$ qubits. We can thus map this state to a PI code with better QEC properties via some unitary. Such a unitary exists by the Knill-Laflamme QEC criterion \cite{Knill_Laflamme_1996}, and can be implemented via GPGs \cite{Johnsson_2020}. Alternatively, we can first utilise our teleportation protocol to map the projected state back to the symmetric space on $N$ qubits, and then apply a unitary to map the resulting state to a PI code on $N$ qubits. On the other hand, if we measure $\hat{J}^2$ and obtain $j = \frac{N-1}{2}$, we can teleport the projected state to the symmetric space on either $N$ or $N+1$ qubits, before mapping to a PI code. Naturally, such flexibility is a useful feature of our protocol.

\section{Discussions \& conclusions}
\label{Sec:Discussions & conclusions}

\noindent In this paper, we provide a general framework to correct single insertion errors on gnu codes via an AM approach. We propose a two-stage syndrome extraction protocol via projective measurements of total AM and its projection along the $z$-axis modulo the code gap $g$, which yields a two-bit syndrome $(j,w)$. We demonstrate that the projected state is mapped to a new quantum code, and we provide an example of our QEC protocol on the four-qubit code \cite{Hagiwara_Nakayama_Apr2020}. Finally, we detail a teleportation protocol to map this state back to a PI code on the desired number of qubits.

\noindent There is plenty of scope for future research on quantum insertion errors. One fruitful avenue is to extend our work to other PI codes, such as Aydin's $(g, m, \delta, \varepsilon)$-codes \cite{Aydin_2024}, Ouyang's $(b,g,m)$-codes \cite{Ouyang_Jing_Brennen_2025} and more. Another avenue is to consider $t$ insertion errors for $t>1$, a venture that is largely absent from the literature. This would increase the difficulty of an already challenging problem, but progress would be significant for synchronisation error QEC.

\section{Acknowledgements}
\label{Sec:Acknowledgements}

\noindent L.B. acknowledges support from EPSRC under grant number EP/W524360/1. Y.O. acknowledges support from EPSRC under grant number EP/W028115/1, and the EPSRC-funded QCI3 Hub (grant number EP/Z53318X/1.)

%% file: supp.tex
\noindent \textbf{Post-insertion state.} Here, we derive an expression for the post-insertion state $\ket{\Psi^a}$ in the sequentially coupled basis \eqref{Coupling path}. In what follows, we consider the insertion positions $a=0$ and $a=1, \ldots, N$ separately. For $a=0$, we have $\ket{\Psi^0} = \ket{\psi_1} \ket{\psi_N}$, that is,
\begin{equation}
    \ket{\Psi^0} = \sum_{k,m} \alpha_m \beta_k \, \ket{1/2, m} \ket{N/2, k-N/2}
\end{equation}
for $0 \leq k \leq N$ and $m=\pm\frac{1}{2}$. By angular momenta coupling laws, we obtain
\begin{equation}
    \ket{\Psi^0} = \sum_{j,k,m} \alpha_m \beta_k C^{j, k+m-\frac{N}{2}}_{\frac{1}{2}, m; \frac{N}{2}, k-\frac{N}{2}} \ket{j, k+m-N/2}_0,
\end{equation}
where $j=\frac{N \pm 1}{2}$. Here, the subscript zero emphasises that $\ket{j, k+m-N/2}_0$ depends on position $a=0$. This basis is non-orthogonal, so we move to the Schur-Weyl basis \cite{Bacon_Chuang_Harrow_2006}
\begin{equation}
    \label{Change of basis}
    \Ket{j, k+m-N/2}_0 = \sum_p d_{0,p} \Ket{j, k+m-N/2}_p,
\end{equation}
where $1 \leq p \leq N$ and $\sum_p |d_{0,p}|^2 = 1$. Thus the post-insertion state is
\begin{equation}
    \label{Post-insertion state for a=0}
    \ket{\Psi^0} = \sum_{j,k,m,p} d_{0,p} \alpha_m \beta_k C^{j, k+m-\frac{N}{2}}_{\frac{1}{2}, m; \frac{N}{2}, k-\frac{N}{2}} \ket{j, k+m-N/2}_p.
\end{equation}
For $a=1, \ldots, N$, we first write the logical state \eqref{N-qubit PI state} in terms of unnormalised Dicke states $\ket{H^N_k} \coloneqq \sqrt{\binom{N}{k}} \, \ket{D^N_k}$
\begin{equation}
    \label{Unnormalised N-qubit state}
    \ket{\psi_N} \, = \sum_k \frac{\beta_k}{\sqrt{\binom{N}{k}}} \ket{H^N_k},
\end{equation}
where $0 \leq k \leq N$. By the well-known Vandermonde identity
\begin{equation}
    \label{Vandermonde breakdown}
    \ket{H^N_k} = \sum_{l \in \mathcal{A}_k} \ket{H^a_l} \ket{H^{N-a}_{k-l}}, 
\end{equation} 
where $\mathcal{A}_k \coloneqq \{ l: \text{max} (0, k+a-N) \leq l \leq \text{min} (a, k) \}$. Replacing the unnormalised Dicke states in \eqref{Vandermonde breakdown} with their normalised counterparts and substituting back into \eqref{Unnormalised N-qubit state} yields
\begin{equation}   
    \label{Decomposed logical state}
    \ket{\psi_N} \, = \sum_{k, l \in \mathcal{A}_k} \beta_{k,l} \ket{D^a_l} \ket{D^{N-a}_{k-l}},
\end{equation}
where
\begin{equation}
    \label{Beta k,l}
    \beta_{k,l} \coloneqq \beta_k \sqrt{\frac{\binom{a}{l} \binom{N-a}{k-l}}{\binom{N}{k}}}.
\end{equation}
Now consider the insertion error \eqref{Single insertion state} on \eqref{Decomposed logical state}. Since $\Ket{\Psi^a} \coloneqq \pi_a \, \big( \Ket{\psi_N} \Ket{\psi_1} \big)$, there exists a permutation $\pi_a \in S_{N+1}$ such that 
\begin{equation}
    \label{Inserted state}
    \ket{\Psi^a} \, = \sum_{k, l \in \mathcal{A}_k} \beta_{k,l} \ket{D^a_l} \ket{\psi_1} \ket{D^{N-a}_{k-l}}.
\end{equation}
We proceed by coupling the states in \eqref{Inserted state} in turn. First
\begin{equation}
    \ket{D^a_l} \ket{\psi_1} = \sum_{j',m} \alpha_m C^{j', l+m-\frac{a}{2}}_{\frac{a}{2}, l-\frac{a}{2}; \frac{1}{2}, m} \Ket{j', l+m-a/2}, 
\end{equation}
where $j' = \frac{a \pm 1}{2}$. Then \eqref{Inserted state} becomes
\begin{equation}
    \Ket{\Psi^a} \, = \sum_{\substack{j', k, m, \\ l \in \mathcal{A}_k}} \alpha_m \beta_{k, l} \, C^{j', l+m-\frac{a}{2}}_{\frac{a}{2}, l-\frac{a}{2}; \frac{1}{2}, m} \Ket{j', l+m-a/2} \Ket{D^{N-a}_{k-l}}.
\end{equation}
Similarly, 
\begin{equation}
    \Ket{j', l+m-a/2} \Ket{D^{N-a}_{k-l}} = \sum_j C^{j, k+m-\frac{N}{2}}_{j', l+m-\frac{a}{2}; \frac{N-a}{2}, k-l-\frac{N-a}{2}} \\ \Ket{j, k+m-N/2}_a,
\end{equation}
where $j = \frac{N \pm 1}{2}$. As before, the subscript $a$ denotes that $\Ket{j, k+m-N/2}_a$ implicitly depends on positions $a=1, \ldots, N$. Now \eqref{Inserted state} becomes
\begin{equation}
    \Ket{\Psi^a} \, = \sum_{\substack{j',j,k,m, \\ l \in \mathcal{A}_k}} \alpha_m \beta_{k, l} \, C^{j', l+m-\frac{a}{2}}_{\frac{a}{2}, l-\frac{a}{2}; \frac{1}{2}, m} C^{j, k+m-\frac{N}{2}}_{j', l+m-\frac{a}{2}; \frac{N-a}{2}, k-l-\frac{N-a}{2}} \\ \Ket{j, k+m-N/2}_a.
\end{equation}
Performing the same change of basis as in \eqref{Change of basis} yields the post-insertion state
\begin{equation}
    \label{Post-insertion state}
    \Ket{\Psi^a} \, = \sum_{\substack{j',j,k,m, \\ p, l \in \mathcal{A}_k}} d_{a,p} \alpha_m \beta_{k, l} \, C^{j', l+m-\frac{a}{2}}_{\frac{a}{2}, l-\frac{a}{2}; \frac{1}{2}, m} C^{j, k+m-\frac{N}{2}}_{j', l+m-\frac{a}{2}; \frac{N-a}{2}, k-l-\frac{N-a}{2}} \\ \Ket{j, k+m-N/2}_p,
\end{equation}
where $\sum_p |d_{a,p}|^2 = 1$.

\noindent \textbf{Measurement of $\hat{J}^2$.} Here, we derive the post-measurement states after measuring $\hat{J}^2$ on $\Ket{\Psi^a}$. For $a=0$, applying the projector \eqref{J^2 projectors} on \eqref{Post-insertion state for a=0} yields the (unnormalised) post-measurement state 
\begin{equation}
    \label{Unnormalised post-measurement state for a=0}
    \ket{\tilde \Psi^0_j} = \sum_p d_{0,p} \sum_m \alpha_m \sum_k \beta_k \, C^{j, k+m-\frac{N}{2}}_{\frac{N}{2}, k-\frac{N}{2}; \frac{1}{2}, m} \Ket{j, k+m-N/2}_p.
\end{equation}
One can directly show that \eqref{Unnormalised post-measurement state for a=0} has squared norm
\begin{equation}
    \label{Norm after measuring J^2 for a=0}
    \braket{\tilde \Psi^0_j | \tilde \Psi^0_j} = \sum_{m,m'} \alpha^*_m \alpha_{m'} \sum_k \beta^*_k \beta_{k+m-m'} C^{j, k+m-\frac{N}{2}}_{\frac{N}{2}, k-\frac{N}{2}; \frac{1}{2}, m} C^{j, k+m-\frac{N}{2}}_{\frac{N}{2}, k+m-m'-\frac{N}{2}; \frac{1}{2}, m'},
\end{equation}
where we used that $\sum_p |d_{0,p}|^2 = 1$. On the other hand, for $a=1, \ldots, N$, applying \eqref{J^2 projectors} on \eqref{Post-insertion state} yields the (unnormalised) post-measurement state 
\begin{equation}
    \label{Unnormalised post-measurement state}
    \ket{\tilde \Psi^a_j} \, = \sum_p d_{a,p} \sum_m \alpha_m \sum_k \sum_{l \in \mathcal{A}_k} \beta_{k,l} \, C^{j', l+m-\frac{a}{2}}_{\frac{a}{2}, l-\frac{a}{2}; \frac{1}{2}, m} C^{j, k+m-\frac{N}{2}}_{j', l+m-\frac{a}{2}; \frac{N-a}{2}, k-l-\frac{N-a}{2}} \\ \Ket{j, k+m-N/2}_p.
\end{equation}
In this case, \eqref{Unnormalised post-measurement state} has squared norm 
\begin{multline}
    \label{Norm after measuring J^2}
    \braket{\tilde \Psi^a_j | \tilde \Psi^a_j} = \sum_{m,m'} \alpha^*_m \alpha_{m'} \sum_k \sum_{l \in \mathcal{A}_k} \beta^*_{k,l} C^{j', l+m-\frac{a}{2}}_{\frac{a}{2}, l-\frac{a}{2}; \frac{1}{2}, m} C^{j, k+m-\frac{N}{2}}_{j', l+m-\frac{a}{2}; \frac{N-a}{2}, k-l-\frac{N-a}{2}} \\
    \sum_{l' \in \mathcal{A}_{k+m-m'}} \beta_{k+m-m',l'} C^{j', l'+m'-\frac{a}{2}}_{\frac{a}{2}, l'-\frac{a}{2}; \frac{1}{2}, m'} C^{j, k+m-\frac{N}{2}}_{j', l'+m'-\frac{a}{2}; \frac{N-a}{2}, k+m-m'-l'-\frac{N-a}{2}},
\end{multline}
where again where we used that $\sum_p |d_{a,p}|^2 = 1$.

\noindent \emph{Symmetric space.} For $j=\frac{N+1}{2}$, the insertion position is inconsequential, so we drop the superscript $a$ and work with the simpler expression for the squared norm from \eqref{Norm after measuring J^2 for a=0}. Explicitly, we have
\begin{equation}
    \Braket{\tilde \Psi_\frac{N+1}{2} | \tilde \Psi_\frac{N+1}{2}} = \sum_{m,m'} \alpha^*_m \alpha_{m'} \sum_k \beta^*_k \beta_{k+m-m'} C^{\frac{N+1}{2}, k+m-\frac{N}{2}}_{\frac{N}{2}, k-\frac{N}{2}; \frac{1}{2}, m} C^{\frac{N+1}{2}, k+m-\frac{N}{2}}_{\frac{N}{2}, k+m-m'-\frac{N}{2}; \frac{1}{2}, m'}.
\end{equation}
By O'Hara's theorem \cite[Thm. 1]{O_Hara_2001}, 
\begin{equation}
    C^{\frac{N+1}{2}, k+m-\frac{N}{2}}_{\frac{N}{2}, k-\frac{N}{2}; \frac{1}{2}, m} = \sqrt{\frac{\binom{N}{k}}{\binom{N+1}{k+m+\frac{1}{2}}}}, \qquad C^{\frac{N+1}{2}, k+m-\frac{N}{2}}_{\frac{N}{2}, k+m-m'-\frac{N}{2}; \frac{1}{2}, m'} = \sqrt{\frac{\binom{N}{k+m-m'}}{\binom{N+1}{k+m+\frac{1}{2}}}},
\end{equation}
and thus
\begin{equation}
    \Braket{\tilde \Psi_\frac{N+1}{2} | \tilde \Psi_\frac{N+1}{2}} = \sum_{m,m'} \alpha^*_m \alpha_{m'} \sum_k \beta^*_k \, \beta_{k+m-m'} \frac{\sqrt{\binom{N}{k} \binom{N}{k+m-m'}}}{\binom{N+1}{k+m+\frac{1}{2}}}.
\end{equation}
For gnu codes, coefficients logical codewords are gapped by a distance of $g$; that is, $\beta_k \neq 0$ for $k=gi$ only. Thus reindexing $k \rightarrow gi$ for $0 \leq i \leq nu$ and expanding the sum over $m,m'$ gives
\begin{equation}
    \Braket{\tilde \Psi_\frac{N+1}{2} | \tilde \Psi_\frac{N+1}{2}} = \frac{1}{N+1} \sum_i \left( |\alpha_\frac{1}{2}|^2 |\beta_{gi}|^2 \, (gi+1) \right. \\
    + \left. |\alpha_{-\frac{1}{2}}|^2 |\beta_{gi}|^2 \, (N+1-gi) \right).
\end{equation}
Splitting the above summation over even/odd $i$ and substituting the state coefficients $\alpha_m, \beta_{gi}$ from \eqref{Definition of beta} gives
\begin{multline}
    \label{Norm involving binomial sums}
    \Braket{\tilde \Psi_\frac{N+1}{2} | \tilde \Psi_\frac{N+1}{2}} = \frac{2^{-(n-1)}}{N+1} \left( |v_1|^2 |c_0|^2 \sum_{i \, \text{even}} \binom{n}{i} \, (gi+1) + |v_1|^2 |c_1|^2 \sum_{i \, \text{odd}} \binom{n}{i} \, (gi+1) \right. \\
    + \left. |v_0|^2 |c_0|^2 \sum_{i \, \text{even}} \binom{n}{i} \, (N+1-gi) + |v_0|^2 |c_1|^2 \sum_{i \, \text{odd}} \binom{n}{i} \, (N+1-gi) \right).
\end{multline}
This consists of a linear combination of the well-known binomial sums
\begin{equation}
    \label{Binomial sums}
    \sum_{\substack{i \in \{0,1\} \\ (\text{mod} \, 2)}} \binom{n}{i} = 2^{n-1}, \qquad \sum_{\substack{i \in \{0,1\} \\ (\text{mod} \, 2)}} i \, \binom{n}{i} = n \, 2^{n-2}.
\end{equation}
Using \eqref{Binomial sums} and the fact that $|c_0|^2 + |c_1|^2 = 1$, \eqref{Norm involving binomial sums} reduces to
\begin{equation} 
    \Braket{\tilde \Psi_\frac{N+1}{2} | \tilde \Psi_\frac{N+1}{2}} = |v_0|^2 \left( 1-\frac{1}{2u}\frac{N}{N+1} \right) + |v_1|^2 \left( \frac{1}{N+1} + \frac{1}{2u}\frac{N}{N+1} \right).
\end{equation}
Thus the post-measurement state is given by
\begin{equation}
    \Ket{\Psi_\frac{N+1}{2}} = \frac{\Ket{\tilde \Psi_\frac{N+1}{2}}}{\sqrt{|v_0|^2 \, \gamma_0 + |v_1|^2 \, \gamma_1}},
\end{equation}
where
\begin{equation}
    \label{gamma w,N}
    \gamma_{w} \coloneqq 
    \begin{dcases}
        1 - \frac{1}{2u} \frac{N}{N+1}, & w=0, \\
        \frac{1}{N+1} + \frac{1}{2u} \frac{N}{N+1}, & w=1, \\
        0, & 2 \leq w \leq g-1.
    \end{dcases}
\end{equation}

\noindent \emph{Mixed symmetry space.} For $j=\frac{N-1}{2}$ and $a=0$, \eqref{Norm after measuring J^2 for a=0} becomes
\begin{equation}
    \Braket{\tilde \Psi^0_\frac{N-1}{2} | \tilde \Psi^0_\frac{N-1}{2}} = \sum_{m,m'} \alpha^*_m \alpha_{m'} \sum_k \beta^*_k \beta_{k+m-m'} C^{\frac{N-1}{2}, k+m-\frac{N}{2}}_{\frac{N}{2}, k-\frac{N}{2}; \frac{1}{2}, m} C^{\frac{N-1}{2}, k+m-\frac{N}{2}}_{\frac{N}{2}, k+m-m'-\frac{N}{2}; \frac{1}{2}, m'}.
\end{equation}
Expanding the sum over $m,m'$ and using the fact that $\beta_k \, \beta_{k \pm 1} = 0$ for gnu codes, we see that
\begin{equation}
    \Braket{\tilde \Psi^0_\frac{N-1}{2} | \tilde \Psi^0_\frac{N-1}{2}} = |\alpha_{-\frac{1}{2}}|^2 \sum_k |\beta_k|^2 \left( C^{\frac{N-1}{2}, k-\frac{N+1}{2}}_{\frac{N}{2}, k-\frac{N}{2}; \frac{1}{2}, -\frac{1}{2}} \right)^2 + |\alpha_\frac{1}{2}|^2 \sum_k |\beta_k|^2 \left( C^{\frac{N-1}{2}, k-\frac{N-1}{2}}_{\frac{N}{2}, k-\frac{N}{2}; \frac{1}{2}, \frac{1}{2}} \right)^2.
\end{equation}
By recursion \cite[Eq. (3.369)]{Sakurai_Napolitano_2020}, it can be shown that 
\begin{equation}
    \label{Recursed CG coeffs}
    C^{\frac{N-1}{2}, k-\frac{N+1}{2}}_{\frac{N}{2}, k-\frac{N}{2}; \frac{1}{2}, -\frac{1}{2}} = -\sqrt{\frac{N-k}{N+1}}, \qquad C^{\frac{N-1}{2}, k-\frac{N-1}{2}}_{\frac{N}{2}, k-\frac{N}{2}; \frac{1}{2}, \frac{1}{2}} = \sqrt{\frac{k}{N+1}}.
\end{equation}
Using \eqref{Recursed CG coeffs} and simultaneously reindexing $k \rightarrow gi$ for $0 \leq i \leq nu$, we obtain
\begin{equation}
    \Braket{\tilde \Psi^0_\frac{N-1}{2} | \tilde \Psi^0_\frac{N-1}{2}} = |\alpha_{-\frac{1}{2}}|^2 \sum_i |\beta_{gi}|^2 \, \frac{N-gi}{N+1} + |\alpha_\frac{1}{2}|^2 \sum_i |\beta_{gi}|^2 \, \frac{gi}{N+1}.
\end{equation}
Again, splitting the above summation over even/odd $i$ and substituting the state coefficients $\alpha_m, \beta_{gi}$ from \eqref{Definition of beta} gives
\begin{multline}
    \Braket{\tilde \Psi^0_\frac{N-1}{2} | \tilde \Psi^0_\frac{N-1}{2}} = \frac{2^{-(n-1)}}{N+1} \left( |v_1|^2 |c_0|^2 \sum_{i \, \text{even}} \binom{n}{i} \, (N-gi) + |v_1|^2 |c_1|^2 \sum_{i \, \text{odd}} \binom{n}{i} \, (N-gi) \right. \\
    + \left. |v_0|^2 |c_0|^2 \sum_{i \, \text{even}} \binom{n}{i} \, gi + |v_0|^2 |c_1|^2 \sum_{i \, \text{odd}} \binom{n}{i} \, gi \right).
\end{multline}
Using the binomial sums in \eqref{Binomial sums} the fact that $|c_0|^2 + |c_1|^2 = 1$, we obtain
\begin{equation}
    \Braket{\tilde \Psi^0_\frac{N-1}{2} | \tilde \Psi^0_\frac{N-1}{2}} = |v_0|^2 \left( 1 - \frac{1}{2u} \right) \frac{N}{N+1} + |v_1|^2 \frac{1}{2u} \frac{N}{N+1} 
\end{equation}
Thus the post-measurement state is given by
\begin{equation}
    \Ket{\tilde \Psi^0_\frac{N-1}{2}} = \frac{\Ket{\tilde \Psi^0_\frac{N-1}{2}}}{\sqrt{|v_0|^2 \, \gamma_{N,g-1} + |v_1|^2 \, \gamma_{N,0}}},
\end{equation}
where for $a=1, \ldots, N$,
\begin{equation}
    \label{gamma w,a}
    \gamma_{a,w} \coloneqq 
    \begin{dcases}
        \left( 1-\frac{1}{2u} \right) \frac{a}{a+1}, & w=0, \\
        0, & 1 \leq w \leq g-2, \\
        \frac{1}{2u} \frac{a}{a+1}, & w=g-1.
    \end{dcases}
\end{equation}

\noindent Similarly, for $j=\frac{N-1}{2}$ and $a=1, \ldots, N$, \eqref{Norm after measuring J^2} becomes
\begin{multline}
    \label{Complicated expression for a=1,..,N}
    \Braket{\tilde \Psi^a_\frac{N-1}{2} | \tilde \Psi^a_\frac{N-1}{2}} = |\alpha_\frac{1}{2}|^2 \sum_k \left( \sum_{l \in \mathcal{A}_k} \beta^*_{k,l} C^{\frac{a-1}{2}, l-\frac{a-1}{2}}_{\frac{a}{2}, l-\frac{a}{2}; \frac{1}{2}, \frac{1}{2}} C^{\frac{N-1}{2}, k-\frac{N-1}{2}}_{\frac{a-1}{2}, l-\frac{a-1}{2}; \frac{N-a}{2}, k-l-\frac{N-a}{2}} \right) \left( \sum_{l' \in \mathcal{A}_k} \beta_{k,l'} C^{\frac{a-1}{2}, l'-\frac{a-1}{2}}_{\frac{a}{2}, l'-\frac{a}{2}; \frac{1}{2}, \frac{1}{2}} C^{\frac{N-1}{2}, k-\frac{N-1}{2}}_{\frac{a-1}{2}, l'-\frac{a-1}{2}; \frac{N-a}{2}, k-l'-\frac{N-a}{2}} \right) \\
    + |\alpha_{-\frac{1}{2}}|^2 \sum_k \left( \sum_{l \in \mathcal{A}_k} \beta^*_{k,l} C^{\frac{a-1}{2}, l-\frac{a+1}{2}}_{\frac{a}{2}, l-\frac{a}{2}; \frac{1}{2}, -\frac{1}{2}} C^{\frac{N-1}{2}, k-\frac{N+1}{2}}_{\frac{a-1}{2}, l-\frac{a+1}{2}; \frac{N-a}{2}, k-l-\frac{N-a}{2}} \right) \left( \sum_{l' \in \mathcal{A}_k} \beta_{k,l'} C^{\frac{a-1}{2}, l'-\frac{a+1}{2}}_{\frac{a}{2}, l'-\frac{a}{2}; \frac{1}{2}, -\frac{1}{2}} C^{\frac{N-1}{2}, k-\frac{N+1}{2}}_{\frac{a-1}{2}, l'-\frac{a+1}{2}; \frac{N-a}{2}, k-l'-\frac{N-a}{2}} \right).
\end{multline}
where expanded the sum over $m,m'$ and once more used that $\beta_k \, \beta_{k \pm 1} = 0$. Using O'Hara's theorem \cite[Thm. 1]{O_Hara_2001} and recursion \cite[Eq. (3.369)]{Sakurai_Napolitano_2020} to simplify the CG coefficients in \eqref{Complicated expression for a=1,..,N}, we obtain
\begin{align}
    \notag C^{\frac{a-1}{2}, l-\frac{a-1}{2}}_{\frac{a}{2}, l-\frac{a}{2}; \frac{1}{2}, \frac{1}{2}} = -\sqrt{\frac{a-l}{a+1}}, \qquad &C^{\frac{N-1}{2}, k-\frac{N-1}{2}}_{\frac{a-1}{2}, l-\frac{a-1}{2}; \frac{N-a}{2}, k-l-\frac{N-a}{2}} = \sqrt{\frac{\binom{a-1}{l} \binom{N-a}{k-l}}{\binom{N-1}{k}}}; \\
    C^{\frac{a-1}{2}, l-\frac{a+1}{2}}_{\frac{a}{2}, l-\frac{a}{2}; \frac{1}{2}, -\frac{1}{2}} = \sqrt{\frac{l}{a+1}}, \qquad &C^{\frac{N-1}{2}, k-\frac{N+1}{2}}_{\frac{a-1}{2}, l-\frac{a+1}{2}; \frac{N-a}{2}, k-l-\frac{N-a}{2}} = \sqrt{\frac{\binom{a-1}{l-1} \binom{N-a}{k-l}}{\binom{N-1}{k-1}}},
\end{align}
and likewise for $l \rightarrow l'$. Thus \eqref{Complicated expression for a=1,..,N} becomes
\begin{multline}
    \Braket{\tilde \Psi^a_\frac{N-1}{2} | \tilde \Psi^a_\frac{N-1}{2}} = |\alpha_{-\frac{1}{2}}|^2 \frac{a}{a+1} \sum_k \frac{|\beta_k|^2}{\binom{N-1}{k-1}^2} \frac{k}{N} \left( \sum_{l \in \mathcal{A}_k} \binom{a-1}{l-1} \binom{N-a}{k-l} \right)^2 \\
    + |\alpha_{\frac{1}{2}}|^2 \frac{a}{a+1} \sum_k \frac{|\beta_k|^2}{\binom{N-1}{k}^2} \left( 1-\frac{k}{N} \right) \left( \sum_{l \in \mathcal{A}_k} \binom{a-1}{l} \binom{N-a}{k-l} \right)^2.
\end{multline}
Again, making use of the Chu-Vandermonde identity and reindexing $k \rightarrow gi$ for $0 \leq i \leq nu$ yields 
\begin{equation}
    \Braket{\tilde \Psi^a_\frac{N-1}{2} | \tilde \Psi^a_\frac{N-1}{2}} = |\alpha_{-\frac{1}{2}}|^2 \frac{a}{a+1} \sum_k |\beta_{gi}|^2 \, \frac{gi}{N} + |\alpha_{\frac{1}{2}}|^2 \frac{a}{a+1} \sum_k |\beta_{gi}|^2 \left( 1-\frac{gi}{N} \right).
\end{equation}
Substituting the state coefficients $\alpha_m, \beta_{gi}$ from \eqref{Definition of beta} and using $N=gnu$ gives 
\begin{multline}
    \Braket{\tilde \Psi^a_\frac{N-1}{2} | \tilde \Psi^a_\frac{N-1}{2}} = \frac{2^{-(n-1)}}{a+1} \left( |v_1|^2 |c_0|^2 \sum_{i \, \text{even}} \binom{n}{i} \, \left( 1-\frac{i}{nu} \right) + |v_1|^2 |c_1|^2 \sum_{i \, \text{odd}} \binom{n}{i} \, \left( 1-\frac{i}{nu} \right) \right. \\
    + \left. |v_0|^2 |c_0|^2 \sum_{i \, \text{even}} \binom{n}{i} \, \frac{i}{nu} + |v_0|^2 |c_1|^2 \sum_{i \, \text{odd}} \binom{n}{i} \, \frac{i}{nu} \right).
\end{multline}
As before, evaluating the binomial sums using \eqref{Binomial sums} yields the post-measurement state
\begin{equation}
    \Ket{\Psi^a_\frac{N-1}{2}} = \frac{\Ket{\tilde \Psi^a_\frac{N-1}{2}}}{\sqrt{|v_0|^2 \, \gamma_{a,g-1} + |v_1|^2 \, \gamma_{a,0}}},
\end{equation}
where $\gamma_{a,w}$ is defined as in \eqref{gamma w,a}.

\noindent \textbf{Modular measurement of $\hat{J}^z$.} Here, we provide a proof of Lem. \ref{Norm-preserving lemma}. 

\noindent \emph{Symmetric space.} After measuring $\hat{J}^2$ and obtaining $j = \frac{N+1}{2}$, for all $a=0, \ldots, N$ the post-measurement state is
\begin{equation}
    \label{Symmetric post-measurement state}
    \Ket{\Psi_\frac{N+1}{2}} = \frac{1}{\sqrt{|v_0|^2 \, \gamma_0 + |v_1|^2 \, \gamma_1}} \sum_{k,m} \alpha_m \beta_k \, C^{k+m-\frac{N}{2}}_{k-\frac{N}{2}; m} \Ket{D^{N+1}_{k+m+\frac{1}{2}}},
\end{equation}
where we suppressed the $j$-labels for brevity. For $w=0, \ldots, g-1$, from \eqref{General projector} the projectors onto the symmetric space are given by 
\begin{equation}
    \label{Symmetric projector}
    \mathcal{P}^w_{\frac{N+1}{2}} = \sum_i \Ket{D^{N+1}_{gi+w}} \Bra{D^{N+1}_{gi+w}},
\end{equation}
where $0 \leq i \leq \lfloor \frac{N+1-w}{g} \rfloor$. Projecting \eqref{Symmetric post-measurement state} onto the symmetric space via \eqref{Symmetric projector} yields
\begin{equation}
    \label{Symmetric norms for general w}
    \Ket{\tilde \Psi^w_\frac{N+1}{2}} = \frac{1}{\sqrt{|v_0|^2 \, \gamma_0 + |v_1|^2 \, \gamma_1}} \left( \alpha_{-\frac{1}{2}} \sum_i \beta_{gi+w} C^{gi+w-\frac{N+1}{2}}_{gi+w-\frac{N}{2}; -\frac{1}{2}} \Ket{D^{N+1}_{gi+w}} + \alpha_\frac{1}{2} \sum_i \beta_{gi+w-1} C^{gi+w-\frac{N+1}{2}}_{gi+w-1-\frac{N}{2}; \frac{1}{2}} \Ket{D^{N+1}_{gi+w}} \right).
\end{equation}
We have two cases to consider, $w=0$ and $w=1$, since $\beta_k \neq 0$ for $k=gi$. When $w=0$, \eqref{Symmetric norms for general w} becomes
\begin{equation}
    \Ket{\tilde \Psi^0_\frac{N+1}{2}} = \tilde v_0' \left( c_0 \Ket{\tilde 0^0_{\frac{N+1}{2}}} + c_1 \Ket{\tilde 1^0_{\frac{N+1}{2}}} \right),
\end{equation}
where
\begin{equation}
    \tilde v_0' \coloneqq \frac{v_0}{\sqrt{|v_0|^2 \, \gamma_0 + |v_1|^2 \, \gamma_1}}
\end{equation}
and for $0 \leq i \leq nu$,
\begin{align}
    \label{Unnormalised symmetric logical codewords}
    \notag \Ket{\tilde 0^0_{\frac{N+1}{2}}} &\coloneqq 2^{-{\frac{n-1}{2}}} \sum_{i \, \text{even}} \sqrt{\binom{n}{i}} \, C^{gi-\frac{N+1}{2}}_{gi-\frac{N}{2}; -\frac{1}{2}} \Ket{D^{N+1}_{gi}}, \\ 
    \Ket{\tilde 1^0_{\frac{N+1}{2}}} &\coloneqq 2^{-{\frac{n-1}{2}}} \sum_{i \, \text{odd}} \sqrt{\binom{n}{i}} \, C^{gi-\frac{N+1}{2}}_{gi-\frac{N}{2}; -\frac{1}{2}} \Ket{D^{N+1}_{gi}}.
\end{align}
An analogous result can be obtained for $w=1$.

\noindent \emph{Mixed symmetry space.} After measuring $\hat{J}^2$ and obtaining $j = \frac{N-1}{2}$, for $a=0$ the post-measurement state is
\begin{equation}
    \label{Mixed symmetry post-measurement state}
    \Ket{\Psi^0_\frac{N-1}{2}} = \frac{1}{\sqrt{|v_0|^2 \, \gamma_{N,g-1} + |v_1|^2 \, \gamma_{N,0}}} \sum_{k,m,p} d_{0,p} \alpha_m \beta_k C^{k+m-\frac{N}{2}}_{k-\frac{N}{2}; m} \ket{k+m-N/2}_p.
\end{equation}
For $w=0, \ldots, g-1$, from \eqref{General projector} the projectors onto the mixed symmetry space are
\begin{equation}
    \label{Projectors onto mixed symmetry space}
    \mathcal{P}^w_{\frac{N-1}{2}} = \sum_{i,p} \, \Ket{\, gi+w-\frac{N-1}{2}}_p \Bra{gi+w-\frac{N-1}{2} \,}_p, 
\end{equation}
where $1 \leq p \leq N$ and $0 \leq i \leq \lfloor \frac{N-1-w}{g} \rfloor$. Projecting \eqref{Mixed symmetry post-measurement state} onto the mixed symmetry space via \eqref{Projectors onto mixed symmetry space} yields
\begin{multline}
    \Ket{\tilde \Psi^{0,w}_\frac{N-1}{2}} = \frac{1}{\sqrt{|v_0|^2 \, \gamma_{g-1,N} + |v_1|^2 \, \gamma_{0,N}}} \left( \alpha_{-\frac{1}{2}} \sum_{i,p} d_{a,p} \beta_{gi+w+1} \, C^{gi+w-\frac{N-1}{2}}_{gi+w+1-\frac{N}{2}; -\frac{1}{2}} \Ket{\, gi+w-\frac{N-1}{2}}_p \right. \\
    \left. + \alpha_\frac{1}{2} \sum_{i,p} d_{a,p} \beta_{gi+w}  C^{gi+w-\frac{N-1}{2}}_{gi+w-\frac{N}{2}; \frac{1}{2}} \Ket{\, gi+w-\frac{N-1}{2}}_p \right).
\end{multline}
Again, we have only two cases to consider, $w=0$ and $w=g-1$. When $w=0$,
\begin{equation}
    \Ket{\tilde \Psi^{0,0}_\frac{N-1}{2}} = \tilde v_{0,1} \left( c_0 \Ket{\tilde 0^{0,0}_{\frac{N-1}{2}}} + c_1 \Ket{\tilde 1^{0,0}_{\frac{N-1}{2}}} \right),
\end{equation}
where 
\begin{equation}
    \tilde v_{0,1} \coloneqq \frac{v_1}{\sqrt{|v_0|^2 \, \gamma_{N,g-1} + |v_1|^2 \, \gamma_{N,0}}}
\end{equation}
and for $0 \leq i \leq nu-1$, 
\begin{align}
    \label{Unnormalised mixed symmetry logical codewords for a=0}
    \notag \Ket{\tilde 0^{0,0}_{\frac{N-1}{2}}} &\coloneqq 2^{-{\frac{n-1}{2}}} \sum_{i \, \text{even}, \, p} d_{0,p} \sqrt{\binom{n}{i}} \, C^{gi-\frac{N-1}{2}}_{gi-\frac{N}{2}; \frac{1}{2}} \Ket{\, gi-\frac{N-1}{2}}_p, \\
    \Ket{\tilde 1^{0,0}_{\frac{N-1}{2}}} &\coloneqq 2^{-{\frac{n-1}{2}}} \sum_{i \, \text{odd}, \, p} d_{0,p} \sqrt{\binom{n}{i}} \, C^{gi-\frac{N-1}{2}}_{gi-\frac{N}{2}; \frac{1}{2}} \Ket{\, gi-\frac{N-1}{2}}_p.
\end{align}
The case for $w=g-1$ follows similarly. On the other hand, for $a=1, \ldots, N$ the post-measurement state is
\begin{equation}
    \label{Post-measurement state for a=1,..,N}
    \Ket{\Psi^a_\frac{N-1}{2}} = \frac{1}{\sqrt{|v_0|^2 \, \gamma_{a,g-1} + |v_1|^2 \, \gamma_{a,0}}} \sum_{\substack{k, m, p, \\ l \in \mathcal{A}_k}} d_{a,p} \alpha_m \beta_{k, l} C^{l+m-\frac{a}{2}}_{l-\frac{a}{2}; m} C^{k+m-\frac{N}{2}}_{l+m-\frac{a}{2}; k-l-\frac{N-a}{2}} \\ \Ket{\, k+m-N/2}_p.
\end{equation}
Projecting \eqref{Post-measurement state for a=1,..,N} onto the mixed symmetry space via \eqref{Projectors onto mixed symmetry space} yields
\begin{multline}
    \label{Mixed symmetry projection for a=1,...,N}
    \Ket{\tilde \Psi^{a,w}_\frac{N-1}{2}} = \frac{1}{\sqrt{|v_0|^2 \, \gamma_{a,g-1} + |v_1|^2 \, \gamma_{a,0}}} \left( \alpha_{-\frac{1}{2}} \sum_{\substack{i, p, \\ l \in \mathcal{A}_{gi+w+1}}} d_{a,p} \beta_{gi+w+1, l}  C^{l-\frac{a+1}{2}}_{l-\frac{a}{2}; -\frac{1}{2}} C^{gi+w-\frac{N-1}{2}}_{l-\frac{a+1}{2}; gi+w+1-l-\frac{N-a}{2}} \Ket{\, gi+w-\frac{N-1}{2}}_p \right. \\
    \left. + \alpha_\frac{1}{2} \sum_{\substack{i, p, \\ l \in \mathcal{A}_{gi+w}}} d_{a,p} \beta_{gi+w, l} \, C^{l-\frac{a-1}{2}}_{l-\frac{a}{2}; \frac{1}{2}} C^{gi+w-\frac{N-1}{2}}_{l-\frac{a-1}{2}; gi+w-l-\frac{N-a}{2}} \Ket{\, gi+w-\frac{N-1}{2}}_p \right).
\end{multline}
As before, we have two cases to consider, $w=0$ and $w=g-1$. When $w=0$,
\begin{equation}
    \Ket{\tilde \Psi^{a,0}_\frac{N-1}{2}} = \tilde v_{a,1} \left( c_0 \Ket{\tilde 0^{a,0}_{\frac{N-1}{2}}} + c_1 \Ket{\tilde 1^{a,0}_{\frac{N-1}{2}}} \right),
\end{equation}
where 
\begin{equation}
    \tilde v_{a,1} \coloneqq \frac{v_1}{\sqrt{|v_0|^2 \, \gamma_{a,g-1} + |v_1|^2 \, \gamma_{a,0}}}
\end{equation}
and for $0 \leq i \leq nu-1$, 
\begin{align}
    \label{Unnormalised mixed symmetry logical codewords for a=1,...,N}
    \notag \Ket{\tilde 0^{a,0}_{\frac{N-1}{2}}} &\coloneqq 2^{-{\frac{n-1}{2}}} \sum_{\substack{i \, \text{even}, \\ p, \, l \in \mathcal{A}_{gi}}} d_{a,p} \sqrt{\frac{\binom{n}{i} \binom{a}{l} \binom{N-a}{gi-l}}{\binom{N}{gi}}} \, C^{l-\frac{a-1}{2}}_{l-\frac{a}{2}; \frac{1}{2}} C^{gi-\frac{N-1}{2}}_{l-\frac{a-1}{2}; gi-l-\frac{N-a}{2}} \Ket{\, gi-\frac{N-1}{2}}_p, \\
    \Ket{\tilde 1^{a,0}_{\frac{N-1}{2}}} &\coloneqq 2^{-{\frac{n-1}{2}}} \sum_{\substack{i \, \text{odd}, \\ p, \, l \in \mathcal{A}_{gi}}} d_{a,p} \sqrt{\frac{\binom{n}{i} \binom{a}{l} \binom{N-a}{gi-l}}{\binom{N}{gi}}} \, C^{l-\frac{a-1}{2}}_{l-\frac{a}{2}; \frac{1}{2}} C^{gi-\frac{N-1}{2}}_{l-\frac{a-1}{2}; gi-l-\frac{N-a}{2}} \Ket{\, gi-\frac{N-1}{2}}_p.
\end{align}
The case for $w=g-1$ follows similarly. We conclude with a proof of Lem. \ref{Norm-preserving lemma}.
    
\begin{proof}[Proof of Lem. \ref{Norm-preserving lemma}]
    We have two cases to consider, $j=\frac{N \pm 1}{2}$. We begin with the symmetric space with $j=\frac{N+1}{2}$.

    \noindent \textbf{Case I:} \emph{Symmetric space.} From \eqref{Unnormalised symmetric logical codewords}, for $w=0$ we have
    \begin{equation}
        \label{Symmetric norms}
        \Braket{\tilde x^0_{\frac{N+1}{2}} | \tilde x^0_{\frac{N+1}{2}}} = 2^{-(n-1)} \sum_{i \equiv x \, (\text{mod} \, 2)} \binom{n}{i} \left( C^{gi-\frac{N+1}{2}}_{gi-\frac{N}{2}; -\frac{1}{2}} \right)^2.
    \end{equation}
    For equality between $x=0$ and $x=1$, the norm-preserving condition
    \begin{equation}
        \label{Condition}
        \sum_i \binom{n}{i} (-1)^i \left( C^{gi-\frac{N+1}{2}}_{gi-\frac{N}{2}; -\frac{1}{2}} \right)^2 = 0
    \end{equation}
    must hold. By O'Hara's theorem \cite[Thm. 1]{O_Hara_2001}, we can write
    \begin{equation}
        \label{CG squared}
        \left( C^{gi-\frac{N+1}{2}}_{gi-\frac{N}{2}; -\frac{1}{2}} \right)^2 = 1-\frac{gi}{N+1}.
    \end{equation}
    Thus by \cite[Eq. (12)]{Ouyang_2014}, the condition in \eqref{Condition} holds for all $n \geq 2$ since $1-\frac{gi}{N+1}$ is linear in $i$. Hence $\braket{\tilde x^0_{\frac{N+1}{2}} | \tilde x^0_{\frac{N+1}{2}}}$ are equal for $x \in \{0,1\}$, and the case for $w=1$ follows similarly. We now show an analogous result for the mixed symmetry space.

    \noindent \textbf{Case II:} \emph{Mixed symmetry space.} For $a=0$, from \eqref{Unnormalised mixed symmetry logical codewords for a=0} we have 
    \begin{equation}
        \label{Mixed symmetry norms for a=0}
        \Braket{\tilde x^{0,0}_{\frac{N+1}{2}} | \tilde x^{0,0}_{\frac{N+1}{2}}} = 2^{-(n-1)} \sum_{i \equiv x \, (\text{mod} \, 2)} \binom{n}{i} \left( C^{gi-\frac{N-1}{2}}_{gi-\frac{N}{2}; \frac{1}{2}} \right)^2,
    \end{equation}
    where we used that $\sum_p |d_{a,p}|^2 = 1$. We obtain a very similar norm-preserving condition to \eqref{Condition},
    \begin{equation}
        \label{Condition mixed sym}
        \sum_i \binom{n}{i} (-1)^i \left( C^{gi-\frac{N-1}{2}}_{gi-\frac{N}{2}; \frac{1}{2}} \right)^2 = 0.
    \end{equation}
    By recursion \cite[Eq. (3.369)]{Sakurai_Napolitano_2020}, it can be shown that
    \begin{equation}
        C^{gi-\frac{N-1}{2}}_{gi-\frac{N}{2}; \frac{1}{2}} = -\sqrt{\frac{N-gi}{N+1}},
    \end{equation}
    and so $\Big( C^{gi-\frac{N-1}{2}}_{gi-\frac{N}{2}; \frac{1}{2}} \Big)^2$ is linear in $i$. Thus by \cite[Eq. (12)]{Ouyang_2014}, the condition in \eqref{Condition mixed sym} holds for all $n \geq 2$, and $\braket{\tilde x^{0,0}_{\frac{N+1}{2}} | \tilde x^{0,0}_{\frac{N+1}{2}}}$ are equal for $x \in \{0,1\}$. The case for $w=1$ follows similarly. On the other hand, for $a=1, \ldots, N$ one can use \eqref{Unnormalised mixed symmetry logical codewords for a=1,...,N} to show that
    \begin{equation}
        \label{Unnormalised norms}
        \Braket{\tilde x^{a,0}_{\frac{N-1}{2}} | \tilde x^{a,0}_{\frac{N-1}{2}}} = 2^{-(n-1)} \sum_{i \equiv x \, (\text{mod} \, 2)} \frac{\binom{n}{i}}{\binom{N}{gi}} \, f_a(i),
    \end{equation}
    where
    \begin{equation}
        \label{Norm function}
        f_a(i) \coloneqq \left( \sum_{l \in \mathcal{A}_{gi}} \sqrt{\binom{a}{l}\binom{N-a}{gi-l}} \, C^{l-\frac{a-1}{2}}_{l-\frac{a}{2}; \frac{1}{2}} C^{gi-\frac{N-1}{2}}_{l-\frac{a-1}{2}; gi-l-\frac{N-a}{2}} \right)^2.
    \end{equation}
    By recursion \cite[Eq. (3.369)]{Sakurai_Napolitano_2020}, 
    \begin{equation}   
        \label{CG recursion}
        C^{l-\frac{a-1}{2}}_{l-\frac{a}{2}; \frac{1}{2}} = -\sqrt{\frac{a-l}{a+1}}
    \end{equation}
    and by O'Hara's theorem \cite[Thm. 1]{O_Hara_2001},
    \begin{equation}
        \label{Application of O'Hara}
        C^{gi-\frac{N-1}{2}}_{l-\frac{a-1}{2}; gi-l-\frac{N-a}{2}} = \sqrt{\frac{\binom{a-1}{l}\binom{N-a}{gi-l}}{\binom{N-1}{gi}}}.
    \end{equation}
    Using \eqref{CG recursion}, \eqref{Application of O'Hara} and the trivial identity $\binom{a}{l} = \frac{a}{a-l} \, \binom{a-1}{l}$, \eqref{Norm function} then becomes
    \begin{equation}
        f_a(i) = \frac{a}{a+1} \frac{1}{\binom{N-1}{gi}} \left( \sum_{l \in \mathcal{A}_{gi}} \binom{a-1}{l} \binom{N-a}{gi-l} \right)^2.
    \end{equation}
    By the well-known Chu-Vandermonde identity, 
    \begin{equation}
        \sum_{l \in \mathcal{A}_{gi}} \binom{a-1}{l} \binom{N-a}{gi-l} = \binom{N-1}{gi},
    \end{equation}
    and thus $f_a(i) = \frac{a}{a+1} \,\binom{N-1}{gi}$. Substituting this into \eqref{Unnormalised norms} gives
    \begin{equation}
        \label{Summation}
        \Braket{\tilde x^{a,0}_{\frac{N-1}{2}} | \tilde x^{a,0}_{\frac{N-1}{2}}} = 2^{-(n-1)} \frac{a}{a+1} \sum_{i \equiv x \, (\text{mod} \, 2)} \binom{n}{i} \left( 1 - \frac{i}{nu} \right).
    \end{equation}
    Clearly the summand in \eqref{Summation} is zero for all $n \leq i \leq nu$, so the norm-preserving condition reduces to
    \begin{equation}
        \label{Binomial sum}
        \sum_i \binom{n}{i} (-1)^i \left( 1-\frac{i}{nu} \right) = 0
    \end{equation}
    for $0 \leq i \leq n$. Again, by \cite[Eq. (12)]{Ouyang_2014} this holds for all $n \geq 2$ since $1-\frac{i}{nu}$ is linear in $i$, and thus $\braket{\tilde x^{a,0}_{\frac{N-1}{2}} | \tilde x^{a,0}_{\frac{N-1}{2}}}$ are equal for $x \in \{0,1\}$. The case for $w=g-1$ follows in an identical manner. 
\end{proof}